\documentclass[a4paper,twocolumn,superscriptaddress,11pt,accepted=2019-06-09]{quantumarticle}

\pdfoutput=1
\usepackage[utf8]{inputenc}
\usepackage[english]{babel}
\usepackage[T1]{fontenc}
\usepackage{amsmath}
\usepackage{hyperref}
\usepackage{amsthm}
\usepackage{tikz}
\usepackage{lipsum}
\usepackage{braket}
\usepackage{gensymb}

\theoremstyle{definition}
\newtheorem{definition}{Definition}

\theoremstyle{plain}
\newtheorem{theorem}{Theorem}
\newtheorem{lemma}{Lemma}

\usepackage{amssymb}
\newcommand{\pare}[1]{\left(#1\right)}
\newcommand{\pares}[1]{(#1)}

\usepackage{listings}
\newcommand{\Tr}{\operatorname{Tr}}
\newcommand{\mat}[1]{\tilde{#1}}
\usepackage{stackrel}
\usepackage{atbegshi,picture}

\AtBeginShipout{\AtBeginShipoutUpperLeft{%
  \put(\dimexpr\paperwidth-1cm\relax,-1.5cm){\makebox[0pt][r]{\framebox{LA-UR-19-24675}}}%
}}

\begin{document}

\title{Polynomial Scaling of Numerical Diagonalization of the 1D Transverse Field Ising Model into a Commuting Basis using the Pauli Product Representation}
\date{\today}
\author{Benjamin Commeau}
\affiliation{University of Connecticut, USA}
\email{benjamincommeau2@gmail.com}
\homepage{https://github.com/benjamincommeau2}

\maketitle

\begin{abstract}
We report numerical results on the  diagonalization of 1D transverse field Ising model. Numerical simulations using the Pauli product representation yield diagonalization from 3 spins to 22 spins in the transverse field Ising model with the number of global Jacobi unitary transformations and number of final terms in diagonalized spin z representation both grew polynomial with the number of spins. These results computed on a classical computer show promise in constructing a quantum circuit to simulate diagonalized generic many-particle Hamiltonians using polynomial number of gates. 
\end{abstract}

\tableofcontents
\section{Introduction}
Simulating time evolution of quantum many-particle systems is exponentially difficult using existing classical algorithms. Quantum algorithms show promise in providing the equivalent simulations in polynomial time. Improvements in the Suzuki-Trotter approximation for time evolution on a quantum computer are recent developments \cite{somma2016trotter}. It has been shown that Hamltonians on a quantum computer cannot be time evolved faster than linear time \cite{berry2007efficient}. This provides a theoretical road block for diagonalizing Hamiltonians.

If there are $n$ unitary transformations $U_i$, such that a Hamiltonian $H$  is transformed to be diagonal
\begin{align}
U^\dagger_{n-1} \cdots U^\dagger_0 H U_0 \cdots  U_{n-1} 
=
\sum_{k} D_k + \tilde{\epsilon}
\ ,
\end{align}
where $[D_k,D_j]=0$ and $\tilde{\epsilon}$ is a residue error, then the time evolution of a state $\ket{\psi}$ becomes
\begin{align}
e^{-iHt}\ket{\psi} \approx e^{-i\epsilon t}\prod_{k} e^{-iD_kt}\ket{\psi}
\ , 
\end{align}
where fast forwarding in time becomes faster than linear time, thus suggesting a violation of the no fast forwarding theorem \cite{berry2007efficient}. However, a recent article \cite{atia2017fast} provides a hypothesis that the Hamiltonian could be transformed into diagonal form in polynomial time under certain conditions.

We propose a new classical algorithm which shows numerical results of polynomial scaling of the number of global Jacobi unitary transformations $U_i$ with respect to the number of spins in a 1D transverse field Ising model. This classical algorithm uses a specific version of the matrix product density operators \cite{verstraete2004matrix} called the Pauli products representation, which performs numerical diagonalization by outputting the equivalent spin-z model with the same eigenvalue spectral function. All $U_i$ are simple Jacobi Pauli products which suggests easier translation into quantum circuit unitary gate design. We motivate this research specific to the 1D transverse field Ising model, because it is a well studied model with a known analytic solution. See \cite{verstraete2009quantum} for an example of its use.

Though this algorithm is flexible to start diagonalizing any k-body interacting many-particle Hamiltonian it is not a stand-alone algorithm that extracts meaningful eigenvalue information and must be paired with other methods that can extract meaningful eigenvalue information from the final diagonalized form in the spin-z representation in polynomial time. Instead we show diagonlizing specific examples can be done in polynomial time on a classical computer. At the time of submission, there does not exist a theorem that proves or disproves the existence of classical algorithms that can diagonalize k-body interacting many-particle Hamiltonian into the spin-z representation in polynomial time.

This paper is outlined as follows: First section discusses preliminaries for the Pauli product representation. Second section shows numerical results calibrating the proposed algorithm written in the Pauli products representation and its use to diagonalize the 1D transverse field Ising model for various number of spins. Third section outlines the design of the algorithm written in the Pauli products representation. Fourth section discusses open questions regarding the Pauli products representation for classical algorithms. Fifth section concludes our findings. Final section provides relevant supplementary material in the form of theorems, lemmas, corollaries, and definitions.

\section{Preliminaries}\label{sec:preliminaries}
The identity matrix and the Pauli matrices
\begin{align}
I,Z,X,Y=
\begin{pmatrix}
1 & 0 \\
0 & 1
\end{pmatrix}
,
\begin{pmatrix}
1 & 0 \\
0 & -1 
\end{pmatrix}
,
\begin{pmatrix}
0 & 1 \\
1 & 0 
\end{pmatrix}
,
\begin{pmatrix}
0 & -i \\
i & 0 
\end{pmatrix}
\end{align}
span a space of $2\times 2$ matrices that are fundamental in the construction of Pauli product representation. Throughout this paper we refer to Pauli product representation as the gamma representation.

A gamma matrix
\begin{align}
\Gamma^{p,q}
=
[I \otimes X  \otimes Y  \otimes Z  \otimes \cdots]^{p,q}
\end{align}
contains $n$ $2\times 2$ matrices, Pauli matrices $X,Y,Z$ and identity $I$, in sequential Kronecker product, spanning a matrix space $2^n\times 2^n$, where each Pauli matrix and identity matrix is uniquely specified by the binary bits of $p,q=0,\cdots ,2^n-1$. See definition \ref{gamma} for more information. The integers $p,q$ have meaning about the properties of its gamma matrix. The bits of $p$ determine which groups of matrices to pick from: $(I,Z)$ or $(X,Y)$. Notice that each group contains matrices whose non-zero elements are either in the diagonal or off-diagonal of a $2\times 2$ matrix. The bits of $q$ determine whether an alternating sign is included in the chosen matrices. Notice that $(I,X)$ have non-zero matrix elements with no alternating sign, and $(Z, Y)$ have alternating sign.

The chosen name for $\Gamma$ as gamma matrix is inspired by the gamma matrices of Dirac's equation.
For example, using definition \ref{gamma} the gamma matrices from Dirac's equation can be written as $\gamma_0,\gamma_1,\gamma_2,\gamma_3=\Gamma^{0,2},i\Gamma^{2,2},i\Gamma^{3,3},i\Gamma^{2,3}$.

Any Hamiltonian $H$ in a finite dimensional Hilbert space can be expressed as a weighted sum of the gamma-matrices,
\begin{align}
H=\sum_{p,q}h_{p,q} \Gamma^{p,q} \ .
\end{align}
See theorem \ref{gammamatricesdecomposition}. For example
\begin{align}
\begin{pmatrix}
3 & 0 & 7 & 0
\\
0 & 3 & 0 & 1
\\
7 & 0 & 1 & 0
\\
0 & 1 & 0 & 1
\end{pmatrix}
=
(2)\Gamma^{0,0}
+
(1)\Gamma^{0,2}
+
(4)\Gamma^{2,0}
+
(3)\Gamma^{2,1}
\ .
\end{align}

All of the diagonal elements of a matrix are contained in the weights of the $\Gamma^{p,q}$ when $p=0$. See theorem \ref{diagonalelements}. If a matrix is diagonalized in the gamma representation, then all gamma-elements are zero except the diagonal elements, and a one-dimensional Walsh-Hadamard transform is needed to convert the gamma diagonal elements into eigenvalues.

The gamma-matrices form an algebraic group under matrix multiplication. Theorem \ref{mult} shows that two gamma matrices in matrix product $\Gamma^{p,q}$,$\Gamma^{r,s}$ are equal to one gamma-matrix in the same algebraic group $\Gamma^{p\veebar r,q\veebar s}$ times a structure constant $f_{p,q}^{r,s}$ which only has four possible values $1,-1,i,-i$, where $p\veebar r$ is the bitwise exclusive-OR of $p$ and $q$. See definitions \ref{binaryvectorsandfields} and \ref{vectorbinaryfield}.

The gamma matrices are self-similar under tensor products of each other. For example the following tensor product
\begin{align}
\Gamma^{p_1,q_1}
\otimes
\Gamma^{p_2,q_2}
=
\Gamma^{P,Q}
\end{align}
where $P=p_1 \oplus p_2$ and $Q=q_1 \oplus q_2$, and $\oplus$ is the bitwise concatenation operator. This property is useful in representing a large Hilbert space as a concatenation of smaller Hilbert spaces. See theorem \ref{kroneckerproduct} for more information.

The transverse field Ising model \cite{verstraete2009quantum}
\begin{align}
H=\sum_{i=0}^{n-1} \pare{X_iX_{i+1} + 2 Z_i} + Y_0Z_1\cdots Z_{n-2}Y_{n-1}
\label{Hxxz}
\end{align}
has the single spin-z diagonal form of
\begin{align}
D=\sum_{k} \omega_k Z_k 
\label{Dkz}
\end{align}
as shown in \cite{verstraete2009quantum}. It should be noted that the diagonal form in the spin-z representation with the fixed spectral function of its eigenvalues is not guaranteed to be unique. See theorem \ref{uniquenessofdiagonalform}. This means equation \ref{Dkz} is not the only spectral function for equation \ref{Hxxz}, and it is possible to have many interacting spin-z terms in the diagonal form with the same eigenvalues.

Due to the large number of diagonal terms discovered during our simulations, it was not practical to compute all of their eigenvalues to test if the diagonalization was correct. Instead we provide a small example of a Hermitian matrix of 6 random gamma elements and provide numerical evidence of its convergence to its true eigenvalues in the next section.

\section{Results\label{results}}
To test our new classical algorithm using the Pauli products representation we computed the eigenvalues for the sparse gamma matrix listed in table \ref{tblfregnvlcnvrgnc} for calibration.
\begin{table}[ht!]
\begin{tabular}{|c|c|c|}
\hline
$h_{p,q}$ & $p$ & $q$\\
\hline
-0.500231 & 00000111 & 10000011  \\
\hline
0.957786  & 00111010 & 00111100  \\
\hline
-0.245173 & 10000110 & 11100010  \\
\hline
0.345722  & 10111101 & 00110001  \\
\hline
0.172746  & 11000110 & 01110011  \\
\hline
-0.960913 & 11001110 & 10001111 \\
\hline
\end{tabular}
\caption{The Hermitian matrix used in Fig. \ref{egnvlcnvrgnc}.}
\label{tblfregnvlcnvrgnc}
\end{table}
Figure \ref{egnvlcnvrgnc} illustrates the diagonalization history for table \ref{tblfregnvlcnvrgnc}. The red line shows the relative error convergence of the experimental eigenvalues to the correct eigenvalues using theorem \ref{relativedistanceeigenvaluemeasure}. The green line shows the increase in sparse memory during each diagonalization step. The simulation achieved a relative error, using theorem \ref{relativedistanceeigenvaluemeasure}, of $10^{-6}$ with 16 sparse diagonal elements.
\begin{figure*}[ht!]
{\includegraphics[width=0.999\textwidth]{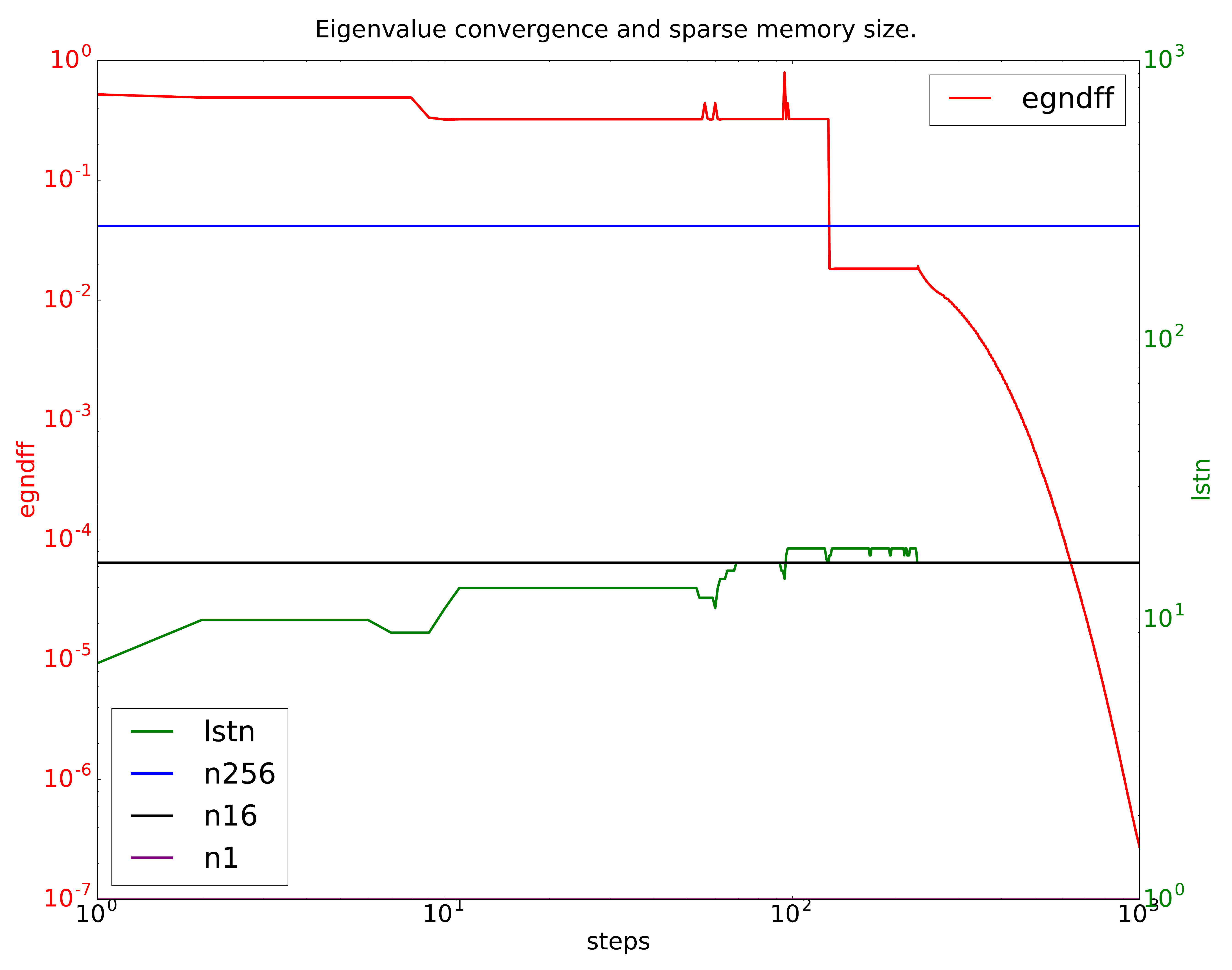}\caption{Hermitian matrix with 6 randomly weighted gamma matrices diagonalized using $\mathcal{S}$ and compared to its true eigenvalues using Python Scipy module. "egndff" in red is the relative difference between the calculated and true eigenvalues. "lstn" are the total number of gamma elements stored in memory at any given moment.}\label{egnvlcnvrgnc}}
\end{figure*}

Note: We used the Python Scipy Coordinate Format (COO) \cite{scipy} to solve the eigenvalues of the sparse matrix in table \ref{tblfregnvlcnvrgnc} using its traditional matrix form. COO does not permit a solution of all the eigenvalues of a sparse matrix, but rather provides 2 less then the total number of eigenvalues, and does not guarantee that the last two eigenvalues are always the smallest in magnitude. This created a problem in comparing the eigenvalues between the two methods. We used a Hungarian algorithm script \cite{macong} to find the best pairwise minimum difference between the two list of eigenvalues, (254) vs. (256), to construct figure \ref{egnvlcnvrgnc}.

\begin{figure*}[ht!]
{\includegraphics[width=0.999\textwidth]{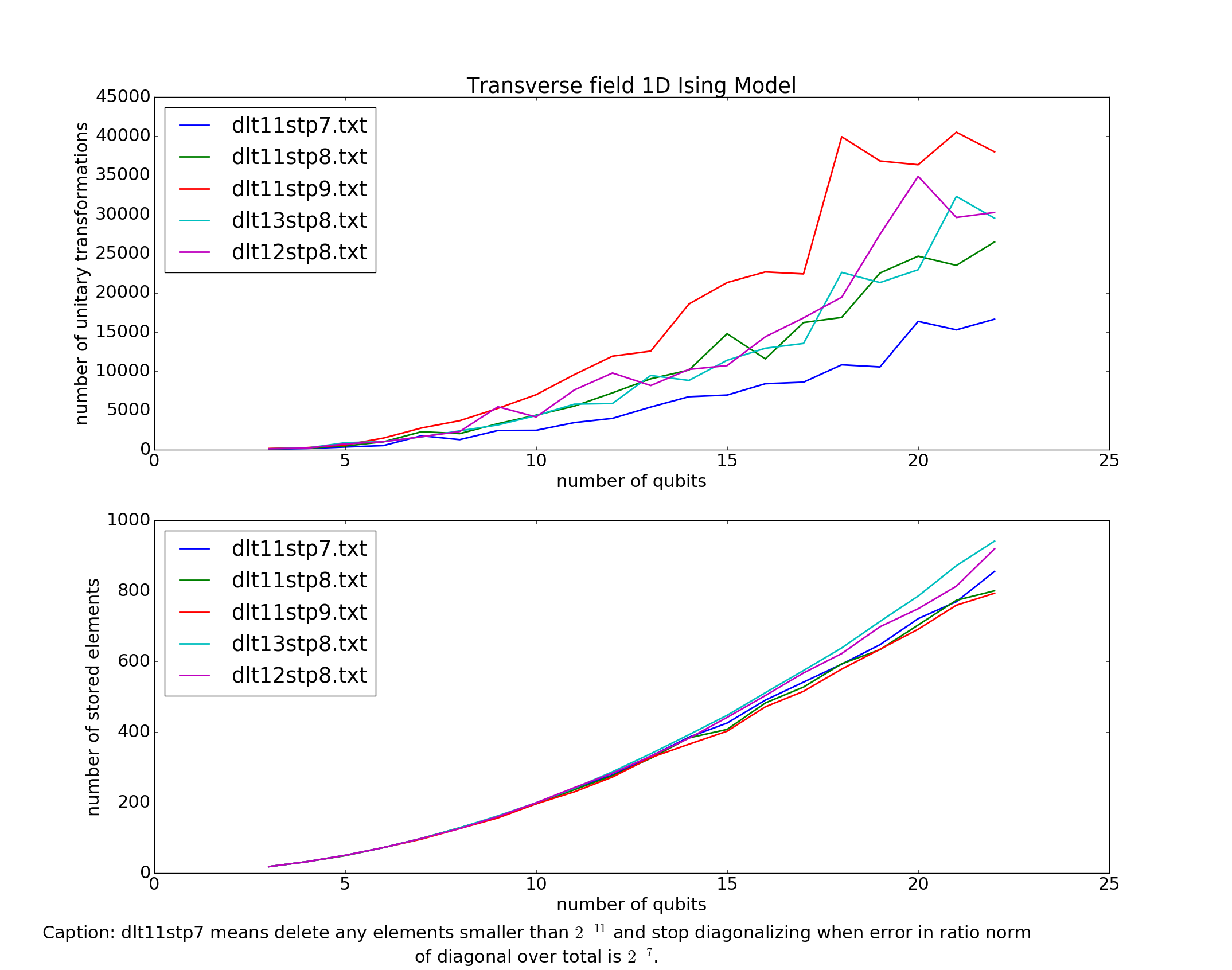}
\caption{Diagonalization for the 1D Transverse Field Ising Model using the Pauli Product representation. dlt11stp7 means delete any elements smaller than $2^{-11}$ and stop diagonalizing when error in ratio norm of diagonal over total is $2^{-7}$.}
\label{transversefieldisingmodelmultidiag}}
\end{figure*}
Figure \ref{transversefieldisingmodelmultidiag} illustrates diagonalization scaling performance of the 1D transverse field Ising model for the number of spins ranging from 3 to 22. The largest computation time on a single CPU task and one of the plotted curves was around 4 hours. There are two error threshold parameters in the simulation: Any gamma element that falls below $|h_{p,q}|\leq \chi$ in magnitude is deleted (dlt) after each diagonalization step. The simulation stops (stp) when $\epsilon$ error is achieved, defined as
\begin{align}
\epsilon
=
1-\dfrac{\sum_{q}(h_{0,q})^2}{\sum_{p,q}(h_{p,q})^2}
\ ,
\end{align}
where $0\leq \epsilon \leq 1$. Note: From the preliminary section $\tilde{\epsilon}\neq \epsilon$ and requires solving a non-trivial equation to translate the two. The following error thresholds were simulated: $(\epsilon,\chi)=$ $(2^{-11}$,$2^{-7})$, $(2^{-11},2^{-8})$, $(2^{-11},2^{-9})$, $(2^{-13},2^{-8})$, $(2^{-12},2^{-8})$. The vertical axis of the top subplot in figure \ref{transversefieldisingmodelmultidiag} labels the total number of unitary transformations to diagonalize the Hamiltonian within the error thresholds, and the vertical axis of the bottom subplot labels the number of gamma elements stored in memory. Both subplots show polynomial scaling with respect to the number of spins in the Hamiltonian.
\section{Algorithm Design}\label{algorithmdesign}
Lemma \ref{unitarytransformationsingleparameter} and \ref{normpreservingrotation} provide the framework on how to diagonalize in the Pauli products representation. Theorem \ref{diagonalizationschemecliffordalgebrarepresentationlocalsingleunitaryparameters} proves that any finite dimensional Hamiltonian can be diagonalized up to some error threshold in the Pauli products representation. The unitary transformation is defined as $U^{r,s}(\phi)=\cos(\phi)I+i \sin(\phi)\Gamma^{r,s}$ and it is a global Jacobi unitary transformation acting on
\begin{align}
H' =\sum_{p,q}h_{p,q}'\Gamma^{p,q} = (U^{r,s}(\phi))^\dagger H U^{r,s}(\phi)
\ .
\end{align}
Global means matrix multiplication between two gamma matrices using traditional indices use all indices and not a subspace of those indices. 

The algorithm needs to find the best $\Gamma^{r,s}$ and $\phi$ to maximize the diagonal norm 
\begin{align}
\text{maximize}\Big( \sum_{s}(h_{0,q_0}')^2 \Big)
\end{align}
at each step. The algorithm steps to achieve diagonalilzation in the Pauli product representation are listed as follows:

The first step is to find the largest vector norm of
\begin{align}
\text{find\_maximum}\Big( \sum_{s}(h_{r,q_r})^2 \Big)
\end{align}
by varying $r$. An example of a sequence of gamma elements with fixed $r$ is
\begin{align}
\Gamma^{r,s} = \binom{I}{Z}_{s_{n-1}}^0 \otimes \binom{X}{Y}_{s_{n-2}}^1 \otimes \binom{I}{Z}_{s_{n-3}}^0 \otimes  \cdots
\end{align}
where $r,s$ have $n$ bits, $r=010\cdots$, and $s_{n-1},s_{n-2},s_{n-3},$ are the bits of $s$ defining which of the two Pauli matrices to pick in the specified sequential Kronecker product. Once the largest vector norm is found, then the best choice of $r$ is found.

The second step is to find the largest list of sparse elements for fixed $r$ that either commute (0) or anti-commute (1) with sparse diagonal elements. This is determined by evaluating the equation $q_0\cdot r \mod 2 =$ $(0)$ or $(1)$, where $q_0\cdot r$ is the bitwise inner product, over the sparse diagonal elements indexed by $q_0$.  

If the anti-commuting case is the largest population, then the best $s$ is determined by evaluating the matrix product of all of the diagonal elements with the off-diagonal elements for fixed $r$ and finding the largest common factor of the matrix products. This is evaluated as 
\begin{align}
\text{most\_common}(q_0\veebar q_r)=s
\ ,
\end{align}
where $q_0$ are the indices spanning the diagonal elements, $q_r$ are the indices spanning the off-diagonal elements for fixed $r$, and $\veebar$ is the bitwise exclusive-OR.

If the commuting case is the largest population, then the best $s$ is determined by toggling the bits of $s$ until $s\cdot r\mod 2 = 1$ is satisfied.
\section{Discussion}\label{discussion}
We hypothesize without proof that their are two reasons why this new classical algorithm may bypass exponential complexity theory: First it includes error thresholds to delete small terms less than $\chi$ during each diagonalization step and final diagonalization stops when total error convergense $\epsilon$ has been achieved. Second the final diagonalized form is a spin-z model which is exponentially hard to extract all of its eigenvalues. 

Regarding specific types of Hamiltonians that cannot have polynomial scaling complexity using this method are the Hamiltonians that have unbounded simultaneous particle interactions, because translating those Hamiltonians into the Pauli products representation will have an exponential number of sparse elements to be stored in memory. As of now there are no theorems that prove or disprove the exponential scaling complexity of this algorithm for any non-trivial classes of Hamiltonians. For example can this new method diagonalize all Hamiltonians that have truncated number of simultaneous particle interactions? If not, can we reduce the class of Hamiltonians further until we find ones that do scale polynomially?.

Regarding designing quantum circuits, it is not obvious how to translate this Pauli products representation into a sequence of unitary gates that fully diagonalize a Hamiltonian. It has the advantage that each unitary transformation operation is all simple Pauli products, but translating the Boolean algebraic rules outlined in the section on algorithm design into equivalent quantum circuit logic is an open question.
\section{Conclusion}\label{conclusion}
We have provided a new method using the Pauli products to construct a classical algorithm that diagonalizes a 1D transverse field Ising model for various number of spins. The data suggests polynomial scaling complexity. The diagonalization is in the representation of a commuting basis, but extracting the eigenvalues is NP-hard. This scaling complexity shows promise in finding a quantum circuit that fully diagnoalizes a Hamiltonian in polynomial scaling complexity. In order for this research to be useful for others, it needs to be paired with other methods that can extract meaningful results from the final diagonal form in the spin-z representation in polynomial time.

\section{Acknowledgement}\label{acknowledgement}
I am grateful to Alexander Balatsky, Patrick Coles, Benjamin Villalonga Correa, Lukasz Cincio, Gayanath Fernando, Matthias Geilhufe, Zhoushen Huang, Sidhant Misra, Sanguthevar Rajasekaran, Andrew Tyler Sornborger, Yigit Subasi, and Marc Denis Vuffray for their useful discussions. The High Performance Computing (HPC) cluster using the Badger account at Los Alamos National Laboratory was used to run all simulations. The funding of this research was made possible by the U.S Department of Energy Office of Science Graduate Student Research (SCGSR) program. Special thanks to Avadh Saxena for his review and critique of this paper.




\appendix
\section{Supplementaries}\label{supplementaries}

\subsection{Binary Vectors and Fields}
\begin{definition}
\textbf{Galois Field 2 (Binary Field).} Let $\textit{GF}(2)$ be a binary field satisfying the following axioms: Let $p,q,r\in \textit{GF}(2)$, then $p,q,r\in (0,1)$ are bits. Let $p\veebar q$ be XOR of $p,q$, and  $p\land q$ be AND of $p,q$, then $\veebar,\land : p\times q \rightarrow \textit{GF}(2)$.  
\begin{itemize}
\item \textbf{Addition (logical-exclusive-or, XOR).} 

\begin{tabular}{|l|l|l|}
\hline
$p$ & $q$ & $(p \veebar q)$ \\
\hline
$0$ & $0$ & $0$  \\
$0$ & $1$ & $1$  \\
$1$ & $0$ & $1$  \\
$1$ & $1$ & $0$  \\
\hline
\end{tabular}
\item \textbf{Multiplication (logical-and, AND)}

\begin{tabular}{|l|l|l|}
\hline
$p$ & $q$ & $(p \land q)$ \\
\hline
$0$ & $0$ & $0$  \\
$0$ & $1$ & $0$  \\
$1$ & $0$ & $0$  \\
$1$ & $1$ & $1$  \\
\hline
\end{tabular}
\item \textbf{Associativity.} For addition $p\veebar (q\veebar r)=(p\veebar q)\veebar r$, and multiplication $p\land (q\land r)=(p\land q)\land r$.
\item \textbf{Commutativity.} For addition $(q\veebar r)=(r\veebar q)$, and multiplication $(q\land r)=(r\land q)$.
\item \textbf{Identity.} For addition $p\veebar 0 = p$ and multiplication $p \land 1 = p$.
\item \textbf{Inverse.} For addition $p\veebar p = 0$ and multiplication $p \land p = 1$, only for $p\neq 0$.
\item \textbf{Distributivity.} $p\land(q\veebar r)=(p\land q)\veebar (p\land r)$.
\item \textbf{Addition with Carry.} Let $+,*$ be addition and multiplication for two integers respectively, then $p+q =p \veebar q + 2*(p\land q).$
\end{itemize}
\label{binaryvectorsandfields}
\end{definition}
\begin{definition}
\textbf{Vector Binary Field.} The vector binary field inherits all properties of $GF(2)$, except addition with carry, and gains additional properties as follows: Let $m$ be a positive integer, $\mathcal{B}(m)\in (0,\cdots ,2^m-1)$ be the vector binary field of bit length $m$, $p,q\in \mathcal{B}(m)$, and $p_j,q_j \in \textit{GF}(2)$, then
\begin{align}
p = \sum_{j=0}^{m-1} 2^j p_j
\ , 
\end{align}
$p_j=\text{floor}(p/2^j) \mod 2$, and $\veebar,\land : p\times q \rightarrow \mathcal{B}(m)$.
\begin{itemize}
\item \textbf{Addition.}
\begin{align}
p\veebar q = \sum_{j=0}^{m-1} 2^j(p_j\veebar q_j)
\ .
\end{align}
\item \textbf{Multiplication.}
\begin{align}
p\land q = \sum_{j=0}^{m-1} 2^j(p_j\land q_j)
\ .
\end{align}
\item \textbf{Metric.}
\begin{align}
p\cdot q = \sum_{j=0}^{m-1} (p_j\land q_j)
\ .
\end{align}
\item \textbf{Metric Addition with Carry.} Taken from lemma \ref{xorsum}, let $+,*$ be addition and multiplication of integers respectively, then
\begin{align}
s\cdot p + s\cdot q = s\cdot (p\veebar 	q) + 2*(s\cdot (p\land q))
\ .
\end{align}
\textbf{Caution.} $s\cdot (p \veebar q)\neq s\cdot (p + q)$ is ill-defined, because $+$ implies addition with carry and $p+q$ can have more bits then $s$. 

\textbf{Caution.} $2s\cdot p$ and $(2s)\cdot p$ are ill-defined, because $(2s)$ can have more bits then $p$. However, $2(s\cdot p)=2*(s\cdot p)$ is defined.

\textbf{Note.} Throughout this paper we extend the $(\cdot)$ operation to the case $s\cdot (p \land q) = s\cdot p \cdot q$ for shorthand notation.
\item \textbf{Kronecker Delta.}
\begin{align}
\delta(p)=
\begin{cases}
1, \ p=0\\
0, \ p \neq 0 \\
\end{cases}
\ .
\end{align}
\end{itemize}
\label{vectorbinaryfield}
\end{definition}
\begin{lemma}
\textbf{Bitwise-XOR-Sum.} Let $s,p,q= 0,\cdots ,2^{m-1}-1$ be integers, and $s_i,p_i,q_i=0,1$ be their bits, then $s\cdot (p\veebar q) = s\cdot p +s\cdot q -2  \left( s\cdot (p \land q)\right)$.	
\label{xorsum}
\end{lemma}
\begin{proof}
By bitwise decomposing 
\begin{align}
s\cdot (p\veebar q)
&=&
\sum_{i}^{m} s_i (p_i \veebar q_i)
\ .
\end{align}
By substituting $p_i \veebar q_i = p_i + q_i - 2p_iq_i$,
\begin{align}
s\cdot (p\veebar q)
&=&
\sum_{i}^{m} s_i (p_i + q_i - 2p_iq_i)
\\&=&
s\cdot p + s\cdot q - 2 (s \cdot (p \land q))
\ .
\end{align}
\end{proof}
\begin{lemma}
\textbf{Hadamard Transform Specific Case.} Let $s,k,q=0,\cdots 2^{m}-1$ be integers of equal bit length $m$, then $\sum_{k=0}^{2^m-1}
(-1)^{(s\veebar q)\cdot k } = 2^m \delta(s\veebar q
)$.
\label{Spcfccshdmrdtrnsfrm}
\end{lemma}
\begin{proof}
By expanding integers into bits,
\begin{align}
\sum_{k=0}^{2^m-1}
(-1)^{(s\veebar q)\cdot k }
= \nonumber \\
\left(\sum_{k_0=0}^{1}\sum_{k_1=0}^{1}\cdots \sum_{k_{m-2}=0}^{1} \sum_{k_{m-1}=0}^{1}\right)
(-1)^{\sum_{r=0}^{m-1}(s_r\veebar q_r) k_r }
\ .
\end{align}
By expanding sum of exponents into products and evaluating the sum over $k_r=0,1$,
\begin{align}
\sum_{k=0}^{2^m-1}
(-1)^{(s\veebar q)\cdot k }
&=& 
\prod_{r=0}^{m-1} \sum_{k_{r}=0}^{1}
(-1)^{(s_r\veebar q_r) k_r }
\\ &=&
\prod_{r=0}^{m-1}
\left( 1+ (-1)^{(s_r\veebar q_r) } \right)
\ .
\end{align}
If any of the bits of $s_r\veebar q_r=1$, then the whole sum is zero. The only case when the sum is non-zero is when all $s_r\veebar q_r=0$, thus
\begin{align}
\sum_{k=0}^{2^m-1}
(-1)^{(s\veebar q)\cdot k }
&=&
2^m \prod_{r=0}^{m-1} \delta(s_r\veebar q_r)
\\&=&
2^m \delta(s\veebar q)
\ .
\end{align}
\end{proof}
\begin{lemma}
\textbf{Bitwise CNOT Gate.} Let $p,q=0,\cdots ,2^m-1$, then the map $(p,q)\longleftrightarrow (p,p\veebar q)$ is bijective.
\label{cnot}
\end{lemma}
\begin{proof}
First expand lemma \ref{cnot} into $k$th bits,
\begin{align}
(p_k,q_k)\longleftrightarrow (p_k,p_k\veebar q_k) ,
\end{align}
and construct their two truth tables.

\begin{table}[!h]
\parbox{.49\linewidth}{
\centering
\begin{tabular}{|l|l|c|l|l|}
\hline
$p_k$ & $q_k$ & $\longleftrightarrow$ & $p_k$ & $(p_k \veebar q_k)$ \\
\hline
$0$ & $0$ & & $0$ & $0$ \\
$0$ & $1$ & & $0$ & $1$ \\
$1$ & $0$ & & $1$ & $1$ \\
$1$ & $1$ & & $1$ & $0$ \\
\hline
\end{tabular}
\caption{\label{cnotbin}CNOT Binary}
}
\\
\parbox{.49\linewidth}{
\centering
\begin{tabular}{|l|c|l|}
\hline
$p_k\oplus q_k$ & $\longleftrightarrow$ & $p_k\oplus (p_k \veebar q_k)$ \\
\hline
$0$ & & $0$ \\
$1$ & & $1$ \\
$2$ & & $3$ \\
$3$ & & $2$ \\
\hline
\end{tabular}
\caption{\label{cnotdec}CNOT Decimal}
}
\end{table}

Table \ref{cnotbin} is the binary representation of its truth table and table \ref{cnotdec} is its decimal representation, where $p_k\oplus q_k=2 p_k+q_k$. Table \ref{cnotdec} indicates $(0,1,2,3)\longleftrightarrow(0,1,3,2)$ is bijective, thus $(p_k,q_k)\longleftrightarrow (p_k,p_k\veebar q_k)$ is bijective. Since $(p_k,q_k,p_k\veebar q_k)$ are linearly independent from all $(p_j,q_j,p_j\veebar q_j)$ for $k \neq j$, then $(p,q)\longleftrightarrow (p,p\veebar q)$ is bijective.
\end{proof}

\subsection{$\Gamma$-matrices Decomposition}
\begin{definition}\textbf{$\Gamma$-matrices. }
Let $p,q,i,j$ be positive integers of equal bit length, $p \veebar i$ be bitwise logical-exclusive-OR of $p$ and $i$, $p \cdot i$ be integer count of bits set to one of bitwise AND of $p$ and $q$, $\delta(i)$ be Kronecker delta for bitwise integer operations, where $\delta(i)=1$ when all bits of $i$ are zero, otherwise $\delta(i)=0$, and $\Gamma_{ik}^{pq}$ be a gamma matrix, then $\Gamma_{ij}^{pq} =(-1)^{q\cdot i} (i)^{-p \cdot q} \delta(i\veebar j \veebar p)$.

\textbf{Note.} Throughout this paper we use $i$ in the base of $(i)^{p\cdot q}$ to mean the imaginary number $\sqrt{-1}$ and $i$ in the exponent of $(-1)^{q\cdot i}$ to mean the integer variable $i$.
\label{gamma}
\end{definition}
\begin{theorem}
\textbf{Generating Pauli matrices.} Let $p,q,i,j=0,1$, then using definition \ref{gamma}, $\Gamma^{p,q}$ generates the Pauli matrices $\sigma_x,\sigma_y,\sigma_z$ and $2\times2$ identity I as $\sigma_x=\Gamma^{1,0}$, $\sigma_y=\Gamma^{1,1}$, $\sigma_z=\Gamma^{0,1}$, and $I=\Gamma^{0,0}$.
\label{generatingpaulimatrices}
\end{theorem}
\begin{proof}
The Python script below
\begin{lstlisting}
import numpy as np
def gamma(p,q,i,j):
  res = np.bitwise_xor(p,i)
  res = np.bitwise_xor(res,j)
  if(res==0):
    res=np.bitwise_and(p,q)
    res=res+2*np.bitwise_and(q,i)
    return 1j**(-res)
  else:
    return 0
mat=np.zeros((2,2),dtype=np.complex)
for p in [0,1]:
  for q in [0,1]:
    for i in [0,1]:
      for j in [0,1]:
        mat[i,j]=gamma(p,q,i,j)
    print("p="+str(p)+", q="+str(q))
    print("gamma(p,q)=")
    print(mat)
\end{lstlisting}
will generate all Pauli matrices using definition \ref{gamma} as terminal output of the form
\begin{lstlisting}
p=0, q=0
gamma(p,q)=
[[ 1.+0.j  0.+0.j]
 [ 0.+0.j  1.+0.j]]
p=0, q=1
gamma(p,q)=
[[ 1.+0.j  0.+0.j]
 [ 0.+0.j -1.-0.j]]
p=1, q=0
gamma(p,q)=
[[ 0.+0.j  1.+0.j]
 [ 1.+0.j  0.+0.j]]
p=1, q=1
gamma(p,q)=
[[ 0.+0.j  0.-1.j]
 [-0.+1.j  0.+0.j]]
\end{lstlisting}
where 1.j is the imaginary number and 0.j is imaginary number times zero as used in NumPy.
\end{proof}
\begin{lemma} \textbf{$\Gamma$-matrices Linear Independence.}
Let $p,q,i,j= 0,\cdots,2^m-1$, then there are $2^m\times2^m$ linearly independent matrices $\Gamma_{ij}^{pq}$.
\label{linearlyindependent}
\end{lemma}
\begin{proof}
We show that the space of all complex $2^m\times 2^m\times 2^m\times 2^m$-tensors $\mathcal{N}$ represents a matrix space isomorphism of $\mathbb{C}^{n,n}$, if $n=2^m\times 2^m=2^{2m}$. First, there is a one-to-one map of a tensor $\Gamma\in\mathcal{N}$ to a matrix $\mat{W}\in\mathbb{C}^{n,n}$ mediated by the re-indexing $(p,q)\rightarrow s$, and $(j,k) \rightarrow i$, $\Gamma_{jk}^{pq}\rightarrow \mat{W}_i^{\ s}$, i.e., associating each pair of pairs $((p,q),(j,k))$ to an index pair $(s,i)$. The matrix space structure of $\mathcal{N}$ follows immediately. From Lemma \ref{orth1}, we transform
\begin{align}
\sum_{i=0}^{2^m-1}\sum_{k=0}^{2^m-1}
\Gamma_{ik}^{pq}\Gamma_{ki}^{rs}
=
2^m \delta_{pr}\delta_{qs}
\ ,
\end{align}
into
\begin{align}
\sum_{y=0}^{2^m-1}
(\mat{W}_{y}^{d})^* \mat{W}_{y}^{f}
=
2^m \delta_{df}
\ ,
\end{align}
by using lemma \ref{hermiticity}. Since the mapping of $(p,q)\longleftrightarrow d$, $(i,k)\longleftrightarrow y$, and $(r,s)\longleftrightarrow f$, is one-to-one, and $f,d,y= 0,\cdots,2^{2m}-1$ span the same range of integers, then $2^{-m}(\mat{W}_{y}^{d})^*$ is the matrix inverse of $\mat{W}_{y}^{f}$ and all $\Gamma$s are linearly independent.
\end{proof}

\subsection{$\Gamma$-matrices Properties}
\begin{theorem} \textbf{$\Gamma$-matrices Decomposition.}
Concluding from lemma \ref{linearlyindependent}, any complex $2^m\times 2^m$-dimensional matrix $\mat{A}$ can be expanded in terms of the basis $\mat{\Gamma}^{pq}$, where $p,q=0,\dots,2^m-1$,
\begin{equation}
(\mat{A})_{ik} = \sum_{p=0}^{2^m-1}\sum_{q=0}^{2^m-1} \Gamma_{ik}^{pq}a_{pq} \ ,
\end{equation}
where
\begin{equation}
a_{pq} = 2^{-m}\Tr \left(\mat{\Gamma^{pq}}^\dagger \mat{A}\right) \ .
\end{equation}
\label{gammamatricesdecomposition}
\end{theorem}
\begin{lemma}
\textbf{$\Gamma$-matrices Hermiticity.} Let $\Gamma_{ik}^{pq}$ be a gamma matrix, then $\Gamma_{ji}^{pq}=(\Gamma_{ij}^{pq})^*$.
\label{hermiticity}
\end{lemma}
\begin{proof}
Starting from lemma \ref{gamma}
\begin{align}
(\Gamma_{ij}^{pq})^*
&=& \left((-1)^{q\cdot i} (i)^{-p \cdot q} \delta(i\veebar j \veebar p)\right)^*
\\ &=& (-1)^{q\cdot i} (i)^{p \cdot q} \delta(i\veebar j \veebar p)
\\ &=& (-1)^{q\cdot i} (i)^{2(p \cdot q)-p \cdot q} \delta(i\veebar j \veebar p)
\\ &=& (-1)^{q\cdot i+p \cdot q} (i)^{-p \cdot q} \delta(i\veebar j \veebar p)
\end{align}
Using lemma \ref{xorsum}
\begin{align}
(\Gamma_{ij}^{pq})^* 
 &=& (-1)^{q\cdot (i\veebar p)} (i)^{-p \cdot q} \delta(i\veebar j \veebar p)
\ .
\end{align}
Since $\delta(i\veebar j \veebar p)\neq 0$ only when $j=i \veebar p$,
\begin{align}
(\Gamma_{ij}^{pq})^*
&=& (-1)^{q\cdot j} (i)^{-p \cdot q} \delta(i\veebar j \veebar p)
\\ &=& \Gamma_{ji}^{pq}
\ .
\end{align}
\end{proof}
\begin{lemma}\textbf{$\Gamma$-matrices Orthogonality.}
Let $p,q,i,j= 0,\cdots,2^m-1$, and $\Gamma_{ik}^{pq}$ be a gamma matrix, then $\Tr(\Gamma^{pq}\Gamma^{rs})=2^m \delta_{pr}\delta_{qs}$.
\label{orth1}
\end{lemma}
\begin{proof}
By substituting Def. \ref{gamma} into Lemma \ref{orth1},
\begin{align}
\sum_{i=0}^{2^m-1}\sum_{k=0}^{2^m-1}
\Gamma_{ik}^{pq}\Gamma_{ki}^{rs}
= \nonumber \\ 
\sum_{i=0}^{2^m-1}\sum_{k=0}^{2^m-1}
(-1)^{q\cdot i} (i)^{-p \cdot q} \delta(i\veebar k \veebar p)
\nonumber \\  (-1)^{s\cdot k} (i)^{-r \cdot s} \delta(k\veebar i \veebar r)
\ .
\end{align}
Evaluate the sum over non-zero values for $\delta(i\veebar k \veebar p)$ by substituting $i\Rightarrow k\veebar p$,
\begin{align}
\sum_{i=0}^{2^m-1}\sum_{k=0}^{2^m-1}
\Gamma_{ik}^{pq}\Gamma_{ki}^{rs}
=
\sum_{k=0}^{2^m-1}
(-1)^{q\cdot (k\veebar p)} (i)^{-p \cdot q} 
\nonumber \\  (-1)^{s\cdot k} (i)^{-r \cdot s} \delta(k\veebar (k\veebar p) \veebar r)
\ .
\end{align}
By substituting $k\veebar k = 0$,
\begin{align}
\sum_{i=0}^{2^m-1}\sum_{k=0}^{2^m-1}
\Gamma_{ik}^{pq}\Gamma_{ki}^{rs}
= \nonumber \\
\sum_{k=0}^{2^m-1}
(-1)^{q\cdot (k\veebar p)+s\cdot k} (i)^{-p \cdot q-r \cdot s} \delta(p \veebar r)
\ .
\end{align}
Using lemma \ref{xorsum} to substitute $(-1)^{q\cdot (k\veebar p) +s\cdot k}= (-1)^{(s\veebar q)\cdot k + q\cdot p}$,
\begin{align}
\sum_{i=0}^{2^m-1}\sum_{k=0}^{2^m-1}
\Gamma_{ik}^{pq}\Gamma_{ki}^{rs}
= \nonumber \\
\sum_{k=0}^{2^m-1}
(-1)^{(s\veebar q)\cdot k + q\cdot p} (i)^{-p \cdot q-r \cdot s} \delta(p \veebar r)
\ .
\end{align}
Using lemma \ref{Spcfccshdmrdtrnsfrm} to substitute $\sum_{k=0}^{2^m-1}
(-1)^{(s\veebar q)\cdot k } = 2^m \delta(s\veebar q
)$ and $(-1)^{p\cdot q} = i^{2p\cdot q}$,
\begin{align}
\sum_{i=0}^{2^m-1}\sum_{k=0}^{2^m-1}
\Gamma_{ik}^{pq}\Gamma_{ki}^{rs}
=
 2^m \delta(s\veebar q) 
 (i)^{p \cdot q-r \cdot s} \delta(p \veebar r)
 \ .
\end{align}
By substituting, $p \Rightarrow r$ and $q \Rightarrow s$ into $(i)^{p\cdot q - r\cdot s} = 1$,
\begin{align}
 \sum_{i=0}^{2^m-1}\sum_{k=0}^{2^m-1}
\Gamma_{ik}^{pq}\Gamma_{ki}^{rs}
=
 2^m \delta(s\veebar q) \delta(p \veebar r)
 \\ = 
  2^m \delta_{s q} \delta_{p r}
 \ .
\end{align}
\end{proof}

\begin{theorem}
\textbf{Fast $\Gamma$-matrices Transform.} Let $p,q,i,j=0,\cdots ,2^m -1$, and X,Y be $2^m\times 2^m$ matrices, then $X_{p,q} = 2^{-m} \Tr(\Gamma^{p,q} Y)$
can be computed in $\mathcal{O}(N^2\ln(N))$ steps and using in-place memory, where $N=2^m$.
\label{X2Y}
\end{theorem}
\begin{proof}
We insert Def. \ref{gamma} into theorem \ref{X2Y},
\begin{align}
X_{pq} = 2^{-m} \sum_{i=0}^{2^m-1}\sum_{j=0}^{2^m-1} (-1)^{q\cdot i} (i)^{-p\cdot q} \delta(i\veebar j \veebar p) Y_{ji}
\ .
\end{align}
The term $\delta(i\veebar j \veebar p)$ is zero unless $j=p\veebar i$, thus
\begin{align}
X_{pq} = 2^{-m} \sum_{i=0}^{2^m-1} (-1)^{q\cdot i} (i)^{-p\cdot q}  Y_{p\veebar i,i}
\ .
\end{align}
We factor out terms from the sum not depending on $i$,
\begin{align}
X_{pq} = 2^{-m} (i)^{-p\cdot q}\sum_{i=0}^{2^m-1} (-1)^{q\cdot i}   Y_{p\veebar i,i}
\ .
\end{align}
We define a new matrix $Z_{pi}=Y_{p\veebar i,i}$. From lemma \ref{cnot}, for every index pair $(p,i)$ there is a unique $(p\veebar i,i)$ and vice versa, then all elements in $Y$ can be swapped using in-place memory to obtain $Z$ and consumes $N^2$ number of steps. Thus
\begin{align}
X_{pq} = 2^{-m} (i)^{-p\cdot q}\sum_{i=0}^{2^m-1} (-1)^{q\cdot i}   Z_{pi}
\ .
\end{align}
Each row with fixed $p$ is evaluated independently. The sum over $i$ with kernel function $(-1)^{q\cdot i}$ for each row $p$ is a one-dimensional Walsh-Hadamard transform, which can be computed in $\mathcal{O}(N\ln(N))$ number of steps and using in-place memory \cite{fino1976unified}. Since there are $N$ rows of fixed $p$, the total computation time is $\mathcal{O}(N^2\ln(N))$. 

\textbf{Note.} The xor swap from $Y_{p\veebar i,i}$ to $Z_{pi}$ has a theoretical lower limit of $\mathcal{O}(N^2)$ number of steps. Practical implementations of the xor swap have shown an upper limit of $\mathcal{O}(N^2\log(N))$ steps, because a naive sequential instruction of the xor swap will cause double swapping which will undo the previous attempts to swap elements. A special branching recursive algorithm similar to the fast Walsh-Hadamard transform was required to prevent double swapping of the same elements. The worst column to swap is the last column, because its column index has all of its bits set to one, and since there are $\log(N)$ number of bits in the last column index, it requires $N\log(N)$ number of steps to perform the xor swap. However, all columns before it require fewer steps because their column indices have fewer set bits.
\end{proof}
\begin{theorem}
\textbf{Diagonal elements.} Combining theorem \ref{X2Y} with lemma \ref{orth1}, define the inverse gamma transformation
\begin{align}
Y_{i,j} = \sum_{i=0}^{2^m-1}\sum_{j=0}^{2^m-1} 
\Gamma^{p,q}_{i,j}
X_{p,q}
\ ,
\end{align}
then $X_{0,q}$ contains all information exclusive to the diagonal elements $Y_{i,i}$ and vice versa.
\label{diagonalelements}
\end{theorem}
\begin{proof}
Starting from definition \ref{gamma}
\begin{align}
Y_{i,j} = \sum_{p=0}^{2^m-1}\sum_{q=0}^{2^m-1} 
(-1)^{q\cdot i} (i)^{-p \cdot q} \delta(i\veebar j \veebar p)
X_{p,q}
\ .
\end{align}
By setting $i=j$
\begin{align}
Y_{i,i} = \sum_{p=0}^{2^m-1}\sum_{q=0}^{2^m-1} 
(-1)^{q\cdot i} (i)^{-p \cdot q} \delta(i\veebar i \veebar p)
X_{p,q}
\ .
\end{align}
Since $i\veebar i=0$, then $p=0$ are the only non-zero terms in the sum, thus
\begin{align}
Y_{i,i} = \sum_{q=0}^{2^m-1} 
(-1)^{q\cdot i}
X_{0,q}
\ ,
\end{align}
which is an invertible Walsh-Hadamard transform.
\end{proof}
\begin{theorem}
\textbf{Gamma-matrices Group Multiplication.}
Let $p,q,r,s,i,k,j=0,\cdots ,2^{m}-1$ and $\Gamma$ be a gamma matrix, then
$\Gamma^{pq}\Gamma^{rs}
= 
f_{p,q}^{r,s}
\Gamma^{(p \veebar r),(q\veebar s)}$,
where
$
f_{p,q}^{r,s}=
(i)^{r\cdot q-p\cdot s} 
(-1)^{ (s\veebar q)\cdot(p \land r) +(p\veebar r)\cdot (q\land s)}
$ are the structure constants. $\Gamma^{p,q}\Gamma^{r,s}$ can be computed in $\mathcal{O}(m)$ number of steps.
\label{mult}
\end{theorem}
\begin{proof}
By substituting Def. \ref{gamma} into Lemma \ref{mult},
\begin{align}
\sum_{k=0}^{2^m-1}
\Gamma_{ik}^{pq}\Gamma_{kj}^{rs}
= \nonumber \\ 
\sum_{k=0}^{2^m-1}
(-1)^{q\cdot i} (i)^{-p \cdot q} \delta(i\veebar k \veebar p)
\nonumber \\  (-1)^{s\cdot k} (i)^{-r \cdot s} \delta(k\veebar j \veebar r)
\ .
\end{align}
By evaluating the sum of non-zero terms with $k=p\veebar i$ and combining exponents,
\begin{align}
\sum_{k=0}^{2^m-1}
\Gamma_{ik}^{pq}\Gamma_{kj}^{rs}
= 
(-1)^{q\cdot i + s\cdot (p\veebar i)} 
(i)^{-r \cdot s-p \cdot q} 
\nonumber \\
\delta(i\veebar j\veebar p\veebar r)
\ .
\end{align}
Using lemma \ref{xorsum},
\begin{align}
\sum_{k=0}^{2^m-1}
\Gamma_{ik}^{pq}\Gamma_{kj}^{rs}
= 
(-1)^{(q\veebar s)\cdot i + s\cdot p} 
(i)^{-r \cdot s-p \cdot q} 
\nonumber \\
\delta(i\veebar j\veebar p\veebar r)
\ .
\end{align}
By substituting $\Gamma_{ij}^{(p \veebar r),(q\veebar s)}=(-1)^{(q\veebar s)\cdot i } 
(i)^{-(p \veebar r)\cdot (q\veebar s)} 
\delta(i\veebar j\veebar p\veebar r)$,
\begin{align}
\sum_{k=0}^{2^m-1}
\Gamma_{ik}^{pq}\Gamma_{kj}^{rs}
= 
f_{p,q}^{r,s}
\Gamma_{ij}^{(p \veebar r),(q\veebar s)}
\ ,
\end{align}
where $f_{p,q}^{r,s}=(-1)^{ s\cdot p} (i)^{(p \veebar r)\cdot (q\veebar s)-r \cdot s-p \cdot q} $.

By using lemma \ref{xorsum},
\begin{align}
f_{p,q}^{r,s}=
(-1)^{ s\cdot p} 
\nonumber \\ 
(i)^{(p \veebar r)\cdot q+(p \veebar r)\cdot s-2(p \veebar r)\cdot (q\land s)-r \cdot s-p \cdot q} 
\ .
\end{align}
By using lemma \ref{xorsum} again,
\begin{align}
f_{p,q}^{r,s}=
(-1)^{ s\cdot p} (i)^{-r \cdot s-p \cdot q}
\nonumber \\ 
(i)^{p\cdot q+r\cdot q-2(p \land r)\cdot q+p\cdot s+r\cdot s-2(p \land r)\cdot s} 
\nonumber \\
(i)^{-2(q\land s)\cdot p-2(q\land s)\cdot r+4(p \land r)\cdot (q\land s)}
\ .
\end{align}
By reducing terms and using lemma \ref{xorsum},
\begin{align}
f_{p,q}^{r,s}=
(i)^{r\cdot q-p\cdot s} 
\nonumber \\ 
(-1)^{ (s\veebar q)\cdot(p \land r) +(p\veebar r)\cdot (q\land s)}
\ .
\end{align}
All bitwise operations involving $\cdot,\land, \veebar$ on the integers $p,q,r,s$ can be computed in $\mathcal{O}(m)$ number of steps, because $m$ is the number of their bits.
\end{proof}
\begin{theorem}
\textbf{Gamma-matrices Commutation.} Let $p,q,r,s=0,\cdots ,2^m-1$, then $\Gamma^{pq},\Gamma^{rs}$ commute when $(p\cdot s - q \cdot r )\mod 2 = 0$ and anti-commute when $(p\cdot s - q \cdot r )\mod 2 = 1$.
\label{commutation}
\end{theorem}
\begin{proof}
Using theorem \ref{mult}, express both commuting and anti-commuting cases with $\pm 1$.
\begin{align}
\Gamma^{pq}\Gamma^{rs} \pm \Gamma^{rs}\Gamma^{pq}
= \nonumber \\
f_{p,q}^{r,s}
\Gamma^{(p \veebar r),(q\veebar s)}
\pm
f_{p,q}^{r,s}
\Gamma^{(r\veebar p),(s\veebar q)}
\ .
\end{align}
Using the commutativity of XOR, $(s\veebar q)=(q\veebar s)$,
\begin{align}
\Gamma^{pq}\Gamma^{rs} \pm \Gamma^{rs}\Gamma^{pq}
= \nonumber \\
(f_{p,q}^{r,s}
\pm
f_{r,s}^{p,q})
\Gamma^{(p \veebar r),(q\veebar s)}
\ .
\end{align}
Using the fact $(f_{r,s}^{p,q})\neq 0$,
\begin{align}
\Gamma^{pq}\Gamma^{rs} \pm \Gamma^{rs}\Gamma^{pq}
= \nonumber \\
f_{r,s}^{p,q} \left(\dfrac{f_{p,q}^{r,s}}{f_{r,s}^{p,q}}
\pm
1\right)
\Gamma^{(p \veebar r),(q\veebar s)}
\ .
\label{com1}
\end{align}
Evaluating the fractions separately,
\begin{align}
\dfrac{f_{p,q}^{r,s}}{f_{r,s}^{p,q}}
=
\dfrac{(i)^{r\cdot q-p\cdot s} 
(-1)^{ (s\veebar q)\cdot(p \land r) +(p\veebar r)\cdot (q\land s)}
}{
(i)^{p\cdot s-r\cdot q} 
(-1)^{ (q\veebar s)\cdot(r \land p) +(r\veebar p)\cdot (s\land q)}
}
\end{align}
Using commutativity of the operators $\cdot,\land,\veebar$,
\begin{align}
\dfrac{f_{p,q}^{r,s}}{f_{r,s}^{p,q}}
&=&
(i)^{2(r\cdot q-p\cdot s)}
\\&=&
(-1)^{(r\cdot q-p\cdot s)}
\label{com2}
\end{align}
Combining equations \ref{com1} and \ref{com2},
\begin{align}
\Gamma^{pq}\Gamma^{rs} \pm \Gamma^{rs}\Gamma^{pq}
= \nonumber \\
f_{r,s}^{p,q} \left((-1)^{(r\cdot q-p\cdot s)}
\pm
1\right)
\Gamma^{(p \veebar r),(q\veebar s)}
\ .
\end{align}
Thus $\Gamma^{pq},\Gamma^{rs}$ commute when $(-1)^{(r\cdot q-p\cdot s)}=1$ and anti-commute when $(-1)^{(r\cdot q-p\cdot s)}=-1$. 
\end{proof}
\begin{theorem}
\textbf{Gamma-matrices Kronecker Product.} Let $n,m$ be positive integers, $p,q=0,\cdots ,2^m-1$, $u,v=0,\cdots ,2^n-1$, $p\oplus u =  2^n p+ u=0,\cdots , 2^{n+m}-1$, $a,b=0,\cdots , 2^{n+m}-1$, and $\Gamma$ be a gamma matrix, then $(\Gamma^{pq}\otimes \Gamma^{uv})_{ab}=\Gamma_{ab}^{p\oplus u,q\oplus v}$.
\label{kroneckerproduct}
\end{theorem}
\begin{proof}
Inserting definition \ref{gamma} into theorem \ref{kroneckerproduct},
\begin{align}
(\Gamma^{pq}\otimes \Gamma^{uv})_{ab}
= 
\Gamma^{pq}_{ij} \Gamma^{uv}_{kl}
= \nonumber \\
(-1)^{q\cdot i} (i)^{-p \cdot q} \delta(i\veebar j \veebar p)
\nonumber \\
(-1)^{v\cdot k} (i)^{-u \cdot v} \delta(k\veebar l \veebar u)
\ ,
\end{align}
where $i,j=0,\cdots 2^m-1$, $k,l=0,\cdots , 2^n-1$, $a= i\oplus k$, and $b=j\oplus l$. By combining exponents
\begin{align}
\Gamma^{pq}_{ij} \Gamma^{uv}_{kl}
= \nonumber \\
(-1)^{q\cdot i+v\cdot k} (i)^{-p \cdot q-u \cdot v} \delta(i\veebar j \veebar p)
\delta(k\veebar l \veebar u)
\ .
\end{align}
By linear independence, $q\cdot i+v\cdot k = (q\oplus v)\cdot (i\oplus k)$, $p \cdot q+u \cdot v = (p\oplus u)\cdot (q\oplus v)$, and $\delta(i\veebar j \veebar p)
\delta(k\veebar l \veebar u)=\delta((i\oplus k)\veebar (j\oplus l) \veebar (p\oplus u))$, thus
\begin{align}
\Gamma^{pq}_{ij} \Gamma^{uv}_{kl}
= 
\Gamma^{p\oplus u,q\oplus v}_{i\oplus k,j\oplus l} 
= 
\Gamma^{p\oplus u,q\oplus v}_{a,b} 
\ .
\end{align}
\end{proof}
\begin{lemma}
\textbf{Gamma Triple Product}
Let $n$ be a positive integer, $r,s,p,q,a,b=0,\cdots , 2^n-1$, $f_{rs}^{pq}$ be the structure constants from theorem \ref{mult}, then
\begin{align}
f_{p\veebar a,q\veebar b}^{r,s}=
f_{p,q}^{r,s}
f_{a,b}^{r,s}
g_{p,q,a,b}^{r,s}
\ ,
\end{align}
where
\begin{align}
g_{p,q,a,b}^{r,s}
=
(-1)^{(r\veebar s)\cdot (p\veebar q)\cdot (a\veebar b)
+
r\cdot a\cdot p + s \cdot b \cdot q
} 
\ .
\end{align}
By exchanging two of the three groups $(r,s),(p,q),(a,b)$, $g_{p,q,a,b}^{r,s}$ remains unchanged.
\label{tripleproduct}
\end{lemma}
\begin{proof}
Starting from theorem \ref{mult}
\begin{align}
f_{p\veebar a,q\veebar b}^{r,s}=
(i)^{r\cdot (q\veebar b)-(p\veebar a)\cdot s} 
\nonumber \\ 
(-1)^{ (s\veebar q\veebar b)\cdot((p\veebar a) \land r) +(p\veebar a\veebar r)\cdot ((q\veebar b)\land s)}
\ .
\end{align}
Using lemma \ref{xorsum}
\begin{align}
f_{p\veebar a,q\veebar b}^{r,s}=
(i)^{r\cdot q-p\cdot s} 
(i)^{r\cdot b-a\cdot s} 
(-1)^{r\cdot q\cdot  b+p\cdot a\cdot s} 
\nonumber \\ 
(-1)^{ (s\veebar q\veebar b)\cdot((p\veebar a) \land r) +(p\veebar a\veebar r)\cdot ((q\veebar b)\land s)}
\ .
\end{align}
By distribution 
\begin{align}
f_{p\veebar a,q\veebar b}^{r,s}=
(i)^{r\cdot q-p\cdot s} 
(i)^{r\cdot b-a\cdot s} 
(-1)^{r\cdot q\cdot  b+p\cdot a\cdot s} 
\nonumber \\ 
(-1)^{ (s\veebar q\veebar b)\cdot p \cdot r +(p\veebar a\veebar r)\cdot q\cdot s}
\nonumber \\ 
(-1)^{ (s\veebar q\veebar b)\cdot a \cdot r +(p\veebar a\veebar r)\cdot b\cdot s}
\ .
\end{align}
By distribution 
\begin{align}
f_{p\veebar a,q\veebar b}^{r,s}=
(i)^{r\cdot q-p\cdot s} 
(i)^{r\cdot b-a\cdot s} 
(-1)^{r\cdot q\cdot  b + p\cdot a\cdot s} 
\nonumber \\ 
(-1)^{ (s\veebar q)\cdot p \cdot r + (p\veebar r)\cdot q\cdot s}
\nonumber \\ 
(-1)^{ b\cdot p \cdot r + a\cdot q\cdot s}
\nonumber \\ 
(-1)^{ q\cdot a \cdot r + p\cdot b\cdot s}
\nonumber \\ 
(-1)^{ (s\veebar b)\cdot a \cdot r + (a\veebar r)\cdot b\cdot s}
\ .
\end{align}
Using theorem \ref{mult}
\begin{align}
f_{p\veebar a,q\veebar b}^{r,s}=
f_{p,q}^{r,s}
f_{a,b}^{r,s}
(-1)^{r\cdot q\cdot  b + p\cdot a\cdot s} 
\nonumber \\ 
(-1)^{ b\cdot p \cdot r + a\cdot q\cdot s}
\nonumber \\ 
(-1)^{ q\cdot a \cdot r + p\cdot b\cdot s}
\ .
\end{align}
Using the fact that the remaining six exponents contain triple products of no two terms in the same groups of $(r,s),(p,q),(a,b)$, then expanding the product $(r\veebar s)\cdot (p\veebar q)\cdot (a\veebar b)$ into eight terms must contain six of them excluding two terms not found, $r\cdot a\cdot p$, and $s \cdot b \cdot q$. Thus
\begin{align}
f_{p\veebar a,q\veebar b}^{r,s}=
f_{p,q}^{r,s}
f_{a,b}^{r,s}
g_{p,q,a,b}^{r,s}
\ ,
\end{align}
where
\begin{align}
g_{p,q,a,b}^{r,s}
=
(-1)^{(r\veebar s)\cdot (p\veebar q)\cdot (a\veebar b)
+
r\cdot a\cdot p + s \cdot b \cdot q
} 
\ .
\end{align}
By exchanging two of the three groups $(r,s),(p,q),(a,b)$, $g_{p,q,a,b}^{r,s}$ remains unchanged.
\end{proof}

\subsection{Diagonalization in the Clifford Algebra representation}
\begin{lemma}
\textbf{Unitary transformation of single parameter.} Let $U  = \cos(\phi) I - i \sin(\phi) \Gamma_{r,s}$, $h_{p,q},h_{p,q}'$ be real,
\begin{align}
H=\sum_{p,q}h_{p,q}\Gamma^{p,q}
&&
H'=\sum_{p,q}h_{p,q}'\Gamma^{p,q}
\ ,
\end{align}
and $H'=U(\phi)HU^\dagger(\phi)$, then
\begin{align}
(h')_{p,q}
=
\delta\pare{p\cdot s - q\cdot r \stackrel{2}{=} 0 }  h_{p,q}
\nonumber \\ +  
\delta\pare{p\cdot s - q\cdot r \stackrel{2}{=} 1 }
\nonumber \\ \bigg(
h_{p,q} \cos(2\phi)
+
i\sin(2\phi)  f^{pq}_{rs}h_{p\veebar r,q\veebar s}
\bigg)
\ ,
\end{align}
where $a\stackrel{c}{=}b$ is shorthand for $(a \mod c) = (b \mod c)$.
\label{unitarytransformationsingleparameter}
\end{lemma}
\begin{proof}
Start with
\begin{align}
H'=
\sum_{p,q}h_{p,q}
\pare{\cos(\phi) I - i \sin(\phi) \Gamma^{r,s}}
\Gamma^{p,q}
\nonumber \\
\pare{\cos(\phi) I + i \sin(\phi) \Gamma^{r,s}}
\ .
\end{align}
Expand all terms
\begin{align}
H'=
\sum_{p,q}h_{p,q}
\bigg(
\cos^2(\phi) \Gamma^{p,q}
 +\sin^2(\phi) \Gamma^{r,s}\Gamma^{p,q}
\Gamma^{r,s}
\nonumber \\
+
i \sin(\phi) \cos(\phi)
\pare{
 \Gamma^{p,q}
\Gamma^{r,s}
- \Gamma^{r,s}\Gamma^{p,q}
}
\bigg)
\ .
\end{align}
Using theorem \ref{mult}
and evaluating $\Gamma^{r,s}\Gamma^{p,q}
\Gamma_{r,s}=f_{rs}^{pq}\Gamma_{p\veebar r,q\veebar s}\Gamma_{r,s}
=f_{rs}^{pq}f_{p\veebar r,q\veebar s}^{rs}\Gamma^{p,q}
$,
\begin{align}
H'=
\sum_{p,q}h_{p,q}
\bigg(
\pare{
\cos^2(\phi)
 +\sin^2(\phi) f_{rs}^{pq}f_{p\veebar r,q\veebar s}^{rs}
}
\Gamma^{p,q}
\nonumber \\
+
i \sin(\phi) \cos(\phi)
\pare{
f_{pq}^{rs} 
-
f_{rs}^{pq} 
}
\Gamma^{p\veebar r,q\veebar s}
\bigg)
\ .
\end{align}
Using lemmas \ref{commutation} and \ref{tripleproduct}, evaluate $f_{p\veebar r,q\veebar s}^{rs}=(-1)^{q\cdot r  - p \cdot s}f_{pq}^{rs}$, $f_{pq}^{rs}=(-1)^{q\cdot r  - p \cdot s}f_{rs}^{pq}$, and $(f_{rs}^{pq})^2=(-1)^{q\cdot r  - p \cdot s}$ to obtain
\begin{align}
H'=
\sum_{p,q}h_{p,q}
\bigg(
\pare{
\cos^2(\phi)
 +\sin^2(\phi) (-1)^{q\cdot r  - p \cdot s}
}
\Gamma^{p,q}
\nonumber \\
+
i \sin(\phi) \cos(\phi)
\pare{
1
-
(-1)^{q\cdot r  - p \cdot s}
}
f_{pq}^{rs} 
\Gamma^{p\veebar r,q\veebar s}
\bigg)
\ .
\end{align}
Split the sum and replace $\pare{
1
-
(-1)^{q\cdot r  - p \cdot s}
}$ with $2 \ \delta\pare{q\cdot r  - p \cdot s \stackrel{2}{=}1}$,
\begin{align}
H'=
\sum_{p,q}h_{p,q}
\pare{
\cos^2(\phi)
 +\sin^2(\phi) (-1)^{q\cdot r  - p \cdot s}
}
\Gamma^{p,q}
\nonumber \\
+
\sum_{p,q}^{q\cdot r  - p \cdot s \stackrel{2}{=}1}h_{p,q} \ 
2i \sin(\phi) \cos(\phi)
f_{pq}^{rs} 
\Gamma^{p\veebar r,q\veebar s}
\ ,
\end{align}
where
\begin{align}
\sum_{p,q}^{q\cdot r  - p \cdot s \stackrel{2}{=}1}
=
\sum_{p=0}^{2^n-1}
\sum_{q=0}^{2^n-1}
\delta\pare{
q\cdot r  - p \cdot s \stackrel{2}{=}1
}
\end{align}
is shorthand notation.
Shift the indices $p,q \rightarrow p\veebar r,q\veebar s$ in the second sum and evaluate $(q\veebar s)\cdot r  - (p\veebar r)\cdot s\stackrel{2}{=}q\cdot r  - p\cdot s$,
\begin{align}
H'=
\sum_{p,q}h_{p,q}
\pare{
\cos^2(\phi)
 +\sin^2(\phi) (-1)^{q\cdot r  - p \cdot s}
}
\Gamma^{p,q}
\nonumber \\
+
\sum_{p,q}^{q\cdot r  - p\cdot s \stackrel{2}{=}1}h_{p\veebar r,q\veebar s} \ 
2i \sin(\phi) \cos(\phi)
f_{p\veebar r,q\veebar s}^{rs} 
\Gamma^{p,q}
\ .
\end{align}
Thus, using the trigonometry identities $1=\cos^2(\phi)+\sin^2(\phi)$, $\cos(2\phi)=\cos^2(\phi)-\sin^2(\phi)$, $\sin(2\phi)=2\sin(\phi)\cos(\phi)$ and $f_{p\veebar r,q\veebar s}^{rs} = f_{rs}^{pq} $,
\begin{align}
(h')_{p,q}
=
\delta\pare{p\cdot s - q\cdot r \stackrel{2}{=} 0 }  h_{p,q}
\nonumber \\ +  
\delta\pare{p\cdot s - q\cdot r \stackrel{2}{=} 1 }
\nonumber \\ \bigg(
h_{p,q} \cos(2\phi)
+
i\sin(2\phi)  f^{pq}_{rs}h_{p\veebar r,q\veebar s}
\bigg)
\ .
\end{align}
\end{proof}
\begin{lemma}
\textbf{Norm Preserving Rotation.} 

From lemma \ref{unitarytransformationsingleparameter}, let
\begin{align}
(\mathcal{X}')_{p}^{rs}
=
\sum_{q}^{q\cdot r - p\cdot s \stackrel{2}{=} 1}
\pare{
(h_{p,q}')^2
-
(h_{p\veebar r,q\veebar s}')^2
}
\ ,
\end{align}
and
\begin{align}
(\mathcal{Y}')_{p}^{rs}
=
\sum_{q}^{q\cdot r - p\cdot s \stackrel{2}{=} 1}
2i f^{pq}_{rs}h_{p\veebar r,q\veebar s}'
h_{p,q}'
\ ,
\end{align}
and $(\mathcal{X})_{p}^{rs}$ and $(\mathcal{Y})_{p}^{rs}$ be defined the same except replacing $h'$ with $h$, then $(\mathcal{X}')_{p}^{rs} = \cos(4\phi)\mathcal{X}_{p}^{rs}
+ \sin(4\phi)\mathcal{Y}_{p}^{rs}$ and 
$(\mathcal{Y}')_{p}^{rs} = \cos(4\phi)\mathcal{Y}_{p}^{rs}
- \sin(4\phi)\mathcal{X}_{p}^{rs}$.
\label{normpreservingrotation}
\end{lemma}
\begin{proof}
Starting from lemma \ref{unitarytransformationsingleparameter},  analyze
\begin{align}
\sum_{q}^{q\cdot r - p\cdot s \stackrel{2}{=} 1}
(h_{p,q}')^2
= \nonumber \\
\sum_{q}^{q\cdot r - p\cdot s \stackrel{2}{=} 1}
\bigg(
h_{p,q} \cos(2\phi)
+
i\sin(2\phi)  f^{pq}_{rs}h_{p\veebar r,q\veebar s}
\bigg)^2
\ .
\end{align}
Using $\delta(q\cdot r - p \cdot s \stackrel{2}{=}1)(if^{pq}_{rs})^2 = \delta(q\cdot r - p \cdot s \stackrel{2}{=}1)$, expand the product
\begin{align}
\sum_{q}^{q\cdot r - p\cdot s \stackrel{2}{=} 1}
(h_{p,q}')^2
= \nonumber \\
\sum_{q}^{q\cdot r - p\cdot s \stackrel{2}{=} 1}
\bigg(
h_{p,q}^2 \cos^2(2\phi)
+
\sin^2(2\phi) h_{p\veebar r,q\veebar s}^2
\nonumber \\ +
2i\sin(2\phi) \cos(2\phi) f^{pq}_{rs}h_{p\veebar r,q\veebar s}
h_{p,q} 
\bigg)
\ .
\end{align}
Evaluate
\begin{align}
(\mathcal{X}')_{p}^{rs}
=
\sum_{q}^{q\cdot r - p\cdot s \stackrel{2}{=} 1}
\nonumber \\
\bigg(
h_{p,q}^2 \cos^2(2\phi)
+
\sin^2(2\phi) h_{p\veebar r,q\veebar s}^2
\nonumber \\ +
2i\sin(2\phi) \cos(2\phi) f^{pq}_{rs}h_{p\veebar r,q\veebar s}
h_{p,q} 
\bigg)
\nonumber \\
-
\bigg(
h_{p\veebar r,q\veebar s}^2 \cos^2(2\phi)
+
\sin^2(2\phi) h_{p,q}^2
\nonumber \\ +
2i\sin(2\phi) \cos(2\phi) f^{p\veebar r,q\veebar s}_{rs}
h_{p,q} h_{p\veebar r,q\veebar s}
\bigg)
\ .
\end{align}
Using the trigonometry identities $\cos(4\phi)=\cos^2(2\phi)-\sin^2(2\phi)$, $\sin(4\phi)=2\sin(2\phi)\cos(2\phi)$ and $\delta\pares{q\cdot r - p\cdot s \stackrel{2}{=} 1}$ $f^{p\veebar r,q\veebar s}_{rs}=-\delta\pares{q\cdot r - p\cdot s \stackrel{2}{=} 1}f^{pq}_{rs}$,
\begin{align}
(\mathcal{X}')_{p}^{rs}
=
\sum_{q}^{q\cdot r - p\cdot s \stackrel{2}{=} 1}
\bigg(
\cos(4\phi)
\pare{h_{p,q}^2-h_{p\veebar r,q\veebar s}^2}
\nonumber \\ +
2i\sin(4\phi)  f^{pq}_{rs}h_{p\veebar r,q\veebar s}
h_{p,q} 
\bigg)
\ .
\end{align}
Thus $(\mathcal{X}')_{p}^{rs} = \cos(4\phi)\mathcal{X}_{p}^{rs}
+ \sin(4\phi)\mathcal{Y}_{p}^{rs}$.

Evaluate
\begin{align}
(\mathcal{Y}')_{p}^{rs}
=
\sum_{q}^{q\cdot r - p\cdot s \stackrel{2}{=} 1}
2i f^{pq}_{rs}
\nonumber \\
\bigg(
h_{p,q} \cos(2\phi)
+
i\sin(2\phi)  f^{pq}_{rs}h_{p\veebar r,q\veebar s}
\bigg)
\nonumber \\
\bigg(
h_{p\veebar r,q\veebar s} \cos(2\phi)
+
i\sin(2\phi)  f^{p\veebar r,q\veebar s}_{rs}h_{p,q}
\bigg)
\ .
\end{align}
Expand the product
\begin{align}
(\mathcal{Y}')_{p}^{rs}
=
\sum_{q}^{q\cdot r - p\cdot s \stackrel{2}{=} 1}
2i f^{pq}_{rs}
\bigg(
\nonumber \\
\pare{
\cos^2(2\phi) 
-\sin^2(2\phi)  f^{pq}_{rs}
f^{p\veebar r,q\veebar s}_{rs}
}
h_{p,q}h_{p\veebar r,q\veebar s}
\nonumber \\ + 
i\sin(2\phi)\cos(2\phi)
\pare{
f^{pq}_{rs}
h_{p\veebar r,q\veebar s}^2 
+
f^{p\veebar r,q\veebar s}_{rs}h_{p,q}^2
}
\bigg)
\ .
\end{align}
Using the trigonometry identities $\cos(4\phi)=\cos^2(2\phi)-\sin^2(2\phi)$, $\sin(4\phi)=2\sin(2\phi)\cos(2\phi)$ and $\delta\pares{q\cdot r - p\cdot s \stackrel{2}{=} 1}f^{p\veebar r,q\veebar s}_{rs}=-\delta\pares{q\cdot r - p\cdot s \stackrel{2}{=} 1}f^{pq}_{rs}$ and $\delta\pares{q\cdot r - p\cdot s \stackrel{2}{=} 1}$ $(f^{pq}_{rs})^2=-\delta\pares{q\cdot r - p\cdot s \stackrel{2}{=} 1}$,
\begin{align}
(\mathcal{Y}')_{p}^{rs}
=
\sum_{q}^{q\cdot r - p\cdot s \stackrel{2}{=} 1}
\bigg(
\cos(4\phi) 
2i f^{pq}_{rs}
h_{p,q}h_{p\veebar r,q\veebar s}
\nonumber \\ + 
\sin(4\phi)
\pare{
h_{p\veebar r,q\veebar s}^2 
-h_{p,q}^2
}
\bigg)
\ .
\end{align}
Thus $(\mathcal{Y}')_{p}^{rs} = \cos(4\phi)\mathcal{Y}_{p}^{rs}
- \sin(4\phi)\mathcal{X}_{p}^{rs}$.
\end{proof}
\begin{theorem}
\textbf{Diagonalization scheme in the Clifford Algebra Representation using Local Single Unitary parameters.}

Let $U  = \cos(\phi) I - i \sin(\phi) \Gamma_{r,s}$ be a single parameterized unitary transformation, $h_{p,q}$, and $h_{p,q}'$ be real,
\begin{align}
H=\sum_{p,q}h_{p,q}\Gamma^{p,q}
&&
H'=\sum_{p,q}h_{p,q}'\Gamma^{p,q}
\ ,
\end{align}
and $\mathcal{S}$ be a local diagonalization scheme that computes $H'=U(\phi)HU^\dagger(\phi)$ for many iterations, where at each iteration of mapping $h_{p,q}\rightarrow h_{p,q}'$, $\phi$ is picked to maximize the diagonal norm
\begin{align}
\sum_{q}(h_{0,q}')^2
\ ,
\end{align}
or $\phi=0$ if no improvement of the diagonal norm is possible and $\Gamma_{r,s}$ has unrestricted unique choice for each iteration, then any $H$ can be fully diagonlized up to some error bound using $\mathcal{S}$.
\label{diagonalizationschemecliffordalgebrarepresentationlocalsingleunitaryparameters}
\end{theorem}
\begin{proof}
From theorem \ref{diagonalelements}, it was shown $h_{0,q}$ contains all exclusive information of the diagonal elements, thus maximizing the first row $p=0$ in the gamma representation is equivalent to maximizing the diagonal norm in the traditional matrix representation. Starting from lemmas \ref{unitarytransformationsingleparameter} and \ref{normpreservingrotation}
\begin{align}
\sum_{q}^{q\cdot r - p\cdot s \stackrel{2}{=} 0}
(h_{p,q}')^2
=
\sum_{q}^{q\cdot r - p\cdot s \stackrel{2}{=} 0}
h_{p,q}^2
\ ,
\end{align}
remains unchanged using $U(\phi)$, thus
\begin{align}
(\mathcal{X}')_{p}^{rs}
=
\sum_{q}^{q\cdot r - p\cdot s \stackrel{2}{=} 1}
\pare{
(h_{p,q}')^2
-
(h_{p\veebar r,q\veebar s}')^2
}
\\ = \cos(4\phi)\mathcal{X}_{p}^{rs}
+ \sin(4\phi)\mathcal{Y}_{p}^{rs}
\ ,
\end{align}
and
\begin{align}
(\mathcal{Y}')_{p}^{rs}
=
\sum_{q}^{q\cdot r - p\cdot s \stackrel{2}{=} 1}
2i f^{pq}_{rs}h_{p\veebar r,q\veebar s}'
h_{p,q}'
\\= \cos(4\phi)\mathcal{Y}_{p}^{rs}
- \sin(4\phi)\mathcal{X}_{p}^{rs}
\ ,
\end{align}
are the only elements that change when mapping $h_{p,q}\rightarrow h_{p,q}'$. To maximize the diagonal norm, set $p=0$ and find $\phi$ that sets $(\mathcal{X}')_{p=0}^{rs}$ to the largest positive value, because $(\mathcal{X}')_{p=0}^{rs}$ is the difference between the diagonal norm and its off-diagonal counterpart. Since $(\mathcal{X}')_{p}^{rs}$ and $(\mathcal{Y}')_{p}^{rs}$ have their own norm preserved under their respective two-dimensional rotation using the angle $\phi$, if $(\mathcal{X}')_{p=0}^{rs}$ is maximized, then $(\mathcal{Y}')_{p=0}^{rs}$ must be zero.

Since $(\mathcal{Y}')_{p=0}^{rs}$ is the inner product of the diagonal row vector $h_{0,q}$ and off-diagonal row vector $if_{rs}^{0,q} h_{r,q\veebar s}$ summed over $q$ satisfying $q\cdot r \stackrel{2}{=} 1$, then there are only two cases when $(\mathcal{Y}')_{p=0}^{rs}=0$. Either all of $h_{0,q}$ or $h_{r,q\veebar s}$ are zero, or $h_{0,q}$ and $h_{r,q\veebar s}$ are vector orthogonal under the metric $if_{rs}^{0,q}\delta(q\cdot r \stackrel{2}{=} 1)$. If $h_{0,q}$ and $h_{r,q\veebar s}$ are vector orthogonal, then there does not exist a $\phi$ that can further increase the diagonal norm.

If $h_{0,q}$ and $h_{r,q\veebar s}$ are vector orthogonal under the metric $if_{rs}^{0,q}\delta(q\cdot r \stackrel{2}{=} 1)$ and $\mathcal{X}^{r,s}_0$ is negative, then $\phi = \pm 45 \degree$, which performs a swap of the elements $h_{0,q}$ with $h_{r,q\veebar s}$ satisfying $q\cdot r \stackrel{2}{=} 1$.

If $h_{0,q}$ and $h_{r,q\veebar s}$ are vector orthogonal under the metric $if_{rs}^{0,q}\delta(q\cdot r \stackrel{2}{=} 1)$ and $\mathcal{X}^{r,s}_0$ is positive, then more explanation is needed. Let us assume for fixed choice of $r=c$, that the diagonal row and off-diagonal row $c$ have only one non-zero element $h_{0,q}=h_{0,a}\delta(q=a)$ and $h_{r,q} = h_{r,q} \delta(q=b)$ respectively, then by lemma \ref{cnot} their exists an $s=a \veebar b$ such that $s\veebar b =a$. The condition $q\cdot r \stackrel{2}{=}1 $ does not restrict our choice in $s$ and $f_{r,s}^{0,q}$ does not modify the magnitude of $\mathcal{Y}^{r,s}_{0}$, because there are only two non-zero elements in the inner product. One can extrapolate to more general cases by performing norm preserving rotations on the single non-zero elements in the diagonal and off-diagonal row to construct multiple non-zero values in each row and their vector orthogonality will remain unchanged using the same inner product.

\end{proof}

\begin{theorem} \textbf{Non-uniqueness of diagonal form.}
Let $n$ be a positive integer, 
\begin{align}
D_n = \sum_{k=0}^{2^n-1} d_k \Gamma^{0,k}
\end{align}
be a diagonal matrix with $2^n$ number of eigenvalues $\lambda_k$, and $d_k$ be its diagonal weights, then uniqueness of $\Gamma^{0,k}$ and $\pm d_k$ is not guaranteed for the same list of $\lambda_k$ and follows table \ref{tableuniquenessdiagonal}.

\begin{table}[!h]
\begin{tabular}{|r|r|r|}
\hline
\text{n} & Is $\Gamma^{0,k}$ unique? & Is $\pm d_k$ unique? \\
\hline
1 & yes & yes \\
\hline
2 & no & yes  \\
\hline
n>2 & no & no  \\
\hline
\end{tabular}
\caption{Uniqueness of diagonal form for various size $n$.}
\label{tableuniquenessdiagonal}
\end{table}
\label{uniquenessofdiagonalform}
\end{theorem}
\begin{proof}
For $n=1$, 
\begin{align}
D_1 = d_0 I + d_1 Z
=
\begin{pmatrix}
d_0+d_1 && 0 \\
0 && d_0-d_1
\end{pmatrix}
\ .
\end{align}
The spectral function has the degree of freedom that the positions of its eigenvalues can change positions in the matrix, thus
\begin{align}
D_1' = d_0 I - d_1 Z
=
\begin{pmatrix}
d_0-d_1 && 0 \\
0 && d_0+d_1
\end{pmatrix}
\ .
\end{align}
Thus for $n=1$, $\Gamma^{0,k}=I,Z$ is unique and $\pm d_k$ is unique.

For $n=2$, extract the eigenvalues of $Z\otimes Z$ using $\text{diag}(Z\otimes Z)=\{(1)_0,(-1)_1,(-1)_2,(1)_3\}$, where $(1)_0$ means the number 1 in position 0. Since its eigenvalues can be swapped into different positions, then swapping position 1 with position 3 yields $\{(1)_0,(1)_1,(-1)_2,(-1)_3\}=\text{diag}(I\otimes Z)$. Since all traceless diagonal members $Z\otimes Z,Z\otimes I,I\otimes Z$ have the same list of eigenvalues $\{1,1,-1,-1\}$ in different orders, then does their exist a swapping of their eigenvalues such that $\text{diag}(\text{?})=\pm\{\_,\_,\_,\_\}$ does not appear in $\Gamma^{0,k}$? There is none, because of the following reasoning: The first choice of $\pm 1$ in position zero is fixed, because $d_k$ are allowed to have scalar multiplication of $\pm 1$ as a degree of freedom, thus $\text{diag}(\text{?})=\pm\{1,\_,\_,\_\}$. The remaining three empty $\_$ have to choose from the list $\{-1,-1,1\}$. Since there are
\begin{align}
\text{diag}(Z\otimes Z) = \{1,-1,-1,1\}
\\
\text{diag}(I\otimes Z) = \{1,1,-1,-1\}
\\
\text{diag}(Z\otimes I) = \{1,-1,1,-1\}
\end{align}
three existing members, and three placement choices for $\{-1,-1,1\}$ in $\text{diag}(\text{?})=\pm\{1,\_,\_,\_\}$, then there does not exist another diagonal member that cannot be represented in terms of one $\Gamma^{0,k}$.

For $n>2$ start with $n=3$. The following example
\begin{align}
\text{diag}(Z\otimes I\otimes I)
=
\{(1)_0,(1)_1,(1)_2,(1)_3,
\nonumber \\
(-1)_4,(-1)_5,(-1)_6,(-1)_7,\}
\end{align}
has a diagonal member that does exist as a single $\Gamma^{0,k}$ by swapping positions 3 with 4.
\begin{align}
\{(1)_0,(1)_1,(1)_2,(-1)_3,
\nonumber \\
(1)_4,(-1)_5,(-1)_6,(-1)_7,\}
=
\text{diag}\Big(
Z\otimes I\otimes I
\nonumber \\
-\tfrac{1}{4} (I+Z)\otimes(I-Z)\otimes(I-Z)
\nonumber \\
+\tfrac{1}{4} (I-Z)\otimes(I+Z)\otimes(I+Z)
\Big)
\ ,
\end{align}
which cannot be reduced into a single $\Gamma^{0,k}$.
\end{proof}
\begin{lemma}
\textbf{Relative Distance Eigenvalue Measure.} Let $\vec{\lambda}(s)$ be a vector of eigenvalues evaluated for steps $s$, and $\vec{\lambda}_0$ be fixed eigenvalues that ideally $\vec{\lambda}(s)$ should converge to, and
\begin{align}
\text{rdm}(\vec{\lambda}(s),\vec{\lambda}_0)
=
\dfrac{|\vec{\lambda}(s)-\vec{\lambda}_0|}
{\sqrt{|\vec{\lambda}(s)|^2+|\vec{\lambda}_0|^2}}
\end{align}
be a relative distance measure function, then $\text{rdm}(\vec{\lambda}(s),\vec{\lambda}_0)$ is scale invariant under equal scaling of its inputs, is uniquely zero when $\vec{\lambda}(s)=\vec{\lambda}_0$, and its output is bounded between $0,\sqrt{2}$.
\label{relativedistanceeigenvaluemeasure}
\end{lemma}
\begin{proof}
By scaling $\vec{\lambda}(s),\vec{\lambda}_0\rightarrow c\vec{\lambda}(s),c\vec{\lambda}_0$,
\begin{align}
\text{rdm}(c\vec{\lambda}(s),c\vec{\lambda}_0)
=
\dfrac{|c\vec{\lambda}(s)-c\vec{\lambda}_0|}
{\sqrt{|c\vec{\lambda}(s)|^2+|c\vec{\lambda}_0|^2}}
\\=
\dfrac{c|\vec{\lambda}(s)-\vec{\lambda}_0|}
{\sqrt{c^2|\vec{\lambda}(s)|^2+c^2|\vec{\lambda}_0|^2}}
\\=
\dfrac{c|\vec{\lambda}(s)-\vec{\lambda}_0|}
{c\sqrt{|\vec{\lambda}(s)|^2+|\vec{\lambda}_0|^2}}
\\=
\text{rdm}(\vec{\lambda}(s),\vec{\lambda}_0)
\ .
\end{align}
By expanding $\vec{\lambda}(s),\vec{\lambda}_0$ into their vector components $\lambda_i(s),(\lambda_0)_i$
\begin{align}
\text{rdm}(\vec{\lambda}(s),\vec{\lambda}_0)
=
\sqrt{
\dfrac{
\sum_{i}(\lambda_i(s)-(\lambda_0)_i)^2
}
{
\sum_{i}[(\lambda_i(s))^2+((\lambda_0)_i)^2]
}
}
\ .
\end{align}
All individual terms in the numerator and denominator are positive definite. The term $(\lambda_i(s)-(\lambda_0)_i)^2$ is only zero when $\lambda_i(s)=(\lambda_0)_i$, thus the numerator is only zero when all $(\lambda_i(s)-(\lambda_0)_i)^2$ are zero, e.i. $\vec{\lambda}(s)=\vec{\lambda}_0$. The denominator can only be zero if all $\lambda_i(s),(\lambda_0)_i$ are zero. Since $\text{rdm}(\vec{\lambda}(s),\vec{\lambda}_0)$ is scale invariant under equaling scaling of its inputs, it will never equal zero if both inputs shrink to zero, but relatively converge to different vectors.

By expanding
\begin{align}
[\text{rdm}(\vec{\lambda}(s),\vec{\lambda}_0)]^2
=
\dfrac{|\vec{\lambda}(s)-\vec{\lambda}_0|^2}
{|\vec{\lambda}(s)|^2+|\vec{\lambda}_0|^2}
\\=
\dfrac{
|\vec{\lambda}(s)|^2+|\vec{\lambda}_0|^2
-2\vec{\lambda}(s)\cdot \vec{\lambda}_0
}
{|\vec{\lambda}(s)|^2+|\vec{\lambda}_0|^2}
\\=
1-
\dfrac{
2\vec{\lambda}(s)\cdot \vec{\lambda}_0
}
{|\vec{\lambda}(s)|^2+|\vec{\lambda}_0|^2}
\end{align}
By substituting $|\vec{\lambda}(s)|=\cos(\psi)r$, $|\vec{\lambda}_0|=\sin(\psi)r$, where $0\leq \psi \leq \pi/2$ and $0\leq r$, and $\vec{\lambda}(s)\cdot \vec{\lambda}_0=\cos(\phi)|\vec{\lambda}(s)||\vec{\lambda}_0|$, then
\begin{align}
[\text{rdm}(\vec{\lambda}(s),\vec{\lambda}_0)]^2
\nonumber \\=
1-
\dfrac{
2 \cos(\phi) \sin(\psi)r\cos(\psi)r
}
{(\sin(\psi)r)^2+(\cos(\psi)r)^2}
\ .
\end{align}
By substituting $2\sin(\psi)\cos(\psi)=\sin(2\psi)$, and $\sin^2(\psi)+\cos^2(\psi)=1$
\begin{align}
[\text{rdm}(\vec{\lambda}(s),\vec{\lambda}_0)]^2
=
1-
\cos(\phi) \sin(2\psi)
\ .
\end{align}
Thus $0 \leq \text{rdm}(\vec{\lambda}(s),\vec{\lambda}_0) \leq \sqrt{2}$.
\end{proof}

\end{document}